\documentclass[runningheads]{llncs}
\usepackage{graphicx, amsmath, amssymb, cryptocode}

\newcommand{\F}{\mathbb{F}}
\newcommand{\FF}{\overline{\mathbb{F}}}

\newcommand{\N}{\mathbb{N}}

\newcommand{\Z}{\mathbb{Z}}

\begin{document}

\title{Genus 2 Supersingular Isogeny Oblivious Transfer}
\author{Rams\`es Fern\`andez-Val\`encia\orcidID{0000-0002-8959-636X}} %
\authorrunning{R. Fern\`andez-Val\`encia} %
\institute{
Eurecat, Centre Tecnol\`ogic de Catalunya, IT Security Unit \\
Grup de Recerca en Nous Models de Ciberseguretat (2017 SGR 01239) \\
72 Bilbao Street, 08005 Barcelona, Catalonia \\
\email{ramses.fernandez@eurecat.org}}

\maketitle

\begin{abstract}
We present an oblivious transfer scheme that extends the proposal made in \cite{Barreto2018} based in supersingular isogenies to the setting of principally polarized supersingular abelian surfaces.
\keywords{Oblivious transfer \and Key exchange \and Abelian surfaces.}
\end{abstract}


\section{Introduction}

Oblivious transfer is considered one of the critical problems in cryptography due the importance of the applications that can be built based on it. In particular, it is possible to prove that oblivious transfer is complete for secure multiparty computation \cite{Kilian1988} therefore, given an implementation of oblivious transfer, it is possible to securely evaluate any polynomial-time computable function without any additional primitive.

In this scheme a sender tries to communicate with a receiver in such a way that the sender sends one of, possibly, many messages to the receiver while remaining oblivious about the information that has been sent.

Among the recent applications of oblivious transfer we highlight blockchain technology. To be precise oblivious transfer plays an central role in the creation of private and verifiable smart contracts. Oblivious transfer can also be utilized for exchange of secrets, private information retrieval, and building protocols for signing contracts.

Currently, there exists some concern about the definition of cryptographic algorithms able to resist attacks using a quantum computer. Among the techniques presumably able to lead to quantum-resistant cryptographic schemes, we find the use of isogenies of supersingular elliptic curves, which has proved to be an interesting solution in the definition of key exchanges \cite{Feo2015}, digital signatures \cite{Feo2018} and also has been explored in the definition of hash functions \cite{Charles2009}, \cite{Tachibana2017} and oblivious transfer protocols \cite{Barreto2018}.

As pointed out by \cite{Castryck2019}, there is a general awareness that many proposals, such as the isogeny-based Diffie-Hellman key exchange, should be generalized to principally polarized abelian surfaces. Among the motivations that support the research based on isogenies of higher genus curves we find the fact that the genus $2$ isogeny graph is much regular than the graph in the genus $1$ setting and this gives the chance to achieve similar security levels with less isogeny computations together with the opportunity to perform better security analysis. Furthermore, as noted in \cite{Takashima2018}, the number of $2$-isogenies between elliptic curves is three whereas abelian surfaces have fifteen $(2,2)$-isogenies, and this fact may improve the security of schemes using supersingular hyperelliptic curves. 

In the literature one finds in \cite{Barreto2018} a proposal for an oblivious transfer based in isogenies of elliptic curves. The present work follows the philosophy of \cite{Castryck2019} and extends the results in \cite{Barreto2018} to a genus $2$ setting and, using the ideas developed in \cite{Flynn2019}, defines and provides the analysis for an oblivious transfer over isogenies of principally polarized abelian surfaces.

The structure of the document is the following: Section 1 reviews the genus $2$ construction of Flynn and Ti \cite{Flynn2019} for a supersingular Diffie-Hellman key exchange which represent the main ingredient, together with the ideas in \cite{Barreto2018}, for our definition for an oblivious transfer over principally polarized supersingular abelian surfaces, which is the defined in Section 2. Finally, Section 3 revolves around the correctness, privacity and indistinguishability of our proposal.

\section{Genus $2$ SIDH}\label{2sidh}

This section reviews the genus $2$ scheme for a SIDH \cite{Flynn2019}. Before exploring the Flynn-Ti proposal, we consider appropriate to begin with a survey of the concepts and ideas that allows us to work with jacobians and Richelot isogenies. Those interested in a more detailed description and also interested in the structure of the graphs in the setting of genus $2$ curves are invited to read \cite{Flynn2019}.  

\subsection{Background on abelian surfaces}

Let $p$ be a prime number and consider $A$ an abelian variety over a field $\F_p$. 
Let $m\in\Z$ be such that $\gcd(m, p) = 1$. If $\hat{A}$ denotes the dual variety of $A$ and we write $\mu_m\left(\FF_p\right)$ for the group of $m$-roots of the unity in $\FF_p$, then the $m$-Weil pairing is defined as

\begin{equation}
\langle \star, \star\rangle_m : A\left(\FF_p\right)[m] \times \hat{A}\left(\FF_p\right)[m] \to \mu_m\left(\FF_p\right).
\end{equation}


\begin{proposition}[Theorem 26.2.3 \cite{Galbraith2012}]
The $m$-Weil pairing is bilinear, alternating and non-degenerate.
\end{proposition}

A subgroup $S\subset A[m]$ is called
\begin{enumerate}
\item proper if for all $1< n \leq m:\, A[n]\subsetneq S$;
\item maximal if $S$ is not contained in any other subgroup of $A[m]$, and
\item $m$-isotropic if the $m$-Weil pairing is trivial when restricted to $S$.
\end{enumerate}

A Richelot isogeny is a $(2,2)$-isogeny between jacobians $J_H$ of genus $2$ curves $H$, that is: its kernel is isomorphic, as a group, to $\F_2\otimes \F_2$ and is maximal isotropic with respect to the $2$-Weil pairing.

Those familiar with isogeny graphs know that $j$-invariants play a central role in the definition of the graph. In the genus $2$ context the role played by the $j$-invariant is replaced with the so called $G_2$-invariants \cite{Cardona2002}. The $G_2$-invariants of a genus $2$ curve are absolute invariants that characterize the isomorphism class of the curve. Let $H$ be a genus $2$ curve and let $J_2$, $J_4$, $J_6$, $J_8$ and $J_{10}$ denote the associated Igusa invariants, then the $G_2$-invariants are defined as follows:
\begin{equation}
\mathcal{G} = (g_1, g_2, g_3) := \left( \frac{J_2^5}{J_{10}}, \frac{J^3_2J_4}{J_{10}}, \frac{J_2^2J_6}{J_{10}} \right).
\end{equation}



It is known (Chapter V, Section 13 \cite{Milne1986}) that an ample divisor $\mathcal{L}$ of an abelian variety $A$ defines an isogeny $\varphi_{\mathcal{L}}:A \to \hat{A}$ called a polarization of $A$. If $\varphi_{\mathcal{L}}$ is an isomorphism, then the polarization is called principal, and $A$ is a principally polarized abelian variety. The degree of a polarization is its degree as an isogeny.

The proposal for a genus $2$ SIDH in \cite{Flynn2019} follows closely the steps taken in \cite{Feo2015}. This is possible, essentially, due to the following results, which allow us to work with jacobians of hyperelliptic curves of genus $2$:

\begin{theorem}[Theorem 1 \cite{Flynn2019}]\label{theorem_1}
Given a prime $p$ and a finite field $\F_p$. If $A$ is a principally polarized abelian surface over $\FF_p$, then:
\begin{enumerate}
\item $A\cong J_H$, where $J_H$ denotes the jacobian of some smooth (hyperelliptic) genus $2$ curve $H$, or
\item $A \cong E_1 \times E_2$ for some elliptic curves $E_1$ and $E_2$. 
\end{enumerate}
\end{theorem}

\begin{proof}
The proof of (1) follows from Theorem 4 \cite{Oort1973} whereas (2) is a direct consequence of Theorem 3.1 \cite{Gonzalez2004}. 
\end{proof}

Another key ingredient is Proposition 1 \cite{Flynn2019} which proves the fact that isogenies with isotropic kernels preserve principal polarizations.

\begin{proposition}[Proposition 1 \cite{Flynn2019}]
Let $H$ be a hyperelliptic curve of genus $2$ over $\F_{p^n}$. Let $K$ be a finite, proper and $\F_{p^n}$-rational subgroup of $J_H(\F_{p^n})$. There exists a principally polarized abelian surface $A$ over $\F_{p^n}$ and an isogeny $\varphi: J_H\to A$ with kernel generated by $K$ if, and only if, $K$ is a maximal $m$-isotropic subgroup of $J_H[m]$ for some positive integer $m$.
\end{proposition}

\begin{proof}
The existence of $A$ follows immediately from Theorem 10.1 (Chapter VII, Section 10 \cite{Milne1986}). In order to prove that it is in fact a principally polarized abelian surface, one defines a polarization $\mu = [\deg(\varphi)]\circ \lambda$ on $J_H$, which is equipped with a principal polarization $\lambda$. One gets a polarization on $J_H/K$ of degree $1$ using Theorem 16.8 and Remark 16.9 (Chapter V, Section 16 \cite{Milne1986}).
\end{proof}

\newpage

\subsection{Set-up}

The set-up of the scheme requires the selection of a prime number $p = 2^n3^m f - 1$, where $f\in\Z$ is small and the exponents $n, m\in\Z$ are such that $2^n \approx 3^m$. 

It is also required to select a base supersingular hyperelliptic curve, which can be chosen using the fact that coverings preserve supersingularity. We pick $H_0: y^2 = x^6 + 1$, which is supersingular as it is the double cover of the elliptic curve $C: y^2 = x^3 +1$, which is supersingular in $\F_p$ for $p \equiv 2\,(\text{mod}\, 3)$. We then use a random sequence of Richelot isogenies to get a random principally polarized supersingular abelian surface which, by Theorem \ref{theorem_1}, is isomorphic to the jacobian of a supersingular hyperelliptic curve $H$.

We end up by selecting bases for the torsion subgroups $J_H[2^n]$ and $J_H[3^m]$ and we do so by following Theorem 6 \cite{Lang2019} and choosing points $\{P_1, P_2, P_3, P_4\}$ for $J_H[2^n]$ and $\{Q_1, Q_2, Q_3, Q_4\}$ for $J_H[3^m]$.

\subsection{First round}

This stage requires Alice to choose a set of secret scalars $\{a_i\}_{i = 1,\dots , 12}$ using the techniques described in Section 3.2 \cite{Flynn2019}, together with an isogeny $\varphi_A : J_H \to J_A$ whose kernel is generated, following Lemmata 2 and 3 and Proposition 2 in \cite{Flynn2019}, by the set

\begin{equation}
K_A = \left\langle \sum_{i=1}^4[a_i]P_i, \sum_{i=5}^8[a_i]P_{i-4}, \sum_{i=9}^{12}[a_i]P_{i-8} \right\rangle.
\end{equation}

Alice computes $\{\varphi_A(Q_j)\}_{j = 1, \dots, 4}$ and sends $(\mathcal{G}(J_A), \{\varphi_A(Q_j)\}_{j = 1,\dots, 4})$ to Bob, where $\mathcal{G}(J_\star)$ denotes the $G_2$-invariants of the curve associated to the jacobian.

At the same time, Bob takes a set of secret scalars $\{b_i\}_{i = 1,\dots, 12}$ together with an isogeny $\varphi_B : J_H \to J_B$ whose kernel is generated by

\begin{equation}
K_B = \left\langle\sum_{i=1}^4[b_i]Q_i, \sum_{i=5}^8[b_i]Q_{i-4}, \sum_{i=9}^{12}[b_i]Q_{i-8} \right\rangle.
\end{equation}

Bob computes $\{\varphi_B(Q_j)\}_{j = 1, \dots, 4}$ and sends the tuple $(\mathcal{G}(J_B), \{\varphi_B(P_j)\}_{j = 1,\dots, 4})$ to Alice.

\subsection{Second round}

Once Alice receives the tuple $(\mathcal{G}(J_B), \{\varphi_B(P_j)\}_{j = 1,\dots, 4})$, she will be able to obtain $J_B$ from $\mathcal{G}(J_B)$ and to compute

\begin{equation}
K_{BA} = \left\langle \sum_{i=1}^4[a_i]\varphi_B(P_i), \sum_{i=5}^8[a_i]\varphi_B(P_{i-4}), \sum_{i=}^{12}[a_i]\varphi_B(P_{i-8}) \right\rangle,
\end{equation}

which generates the kernel of a $(2^n, 2^{n-k}, 2^k)$-isogeny $\varphi'_A: J_B\to J_{BA}$.

Similarly, after receiving Alice's tuple, Bob will be able to obtain $J_A$ and to compute the kernel of a $(2^m, 2^{m-k}, 2^k)$-isogeny $\varphi'_B:J_A\to J_{AB}$ generated by the set

\begin{equation}
K_{AB} = \left\langle \sum_{i=1}^4[b_i]\varphi_A(Q_i), \sum_{i=5}^8[b_i]\varphi_A(Q_{i-4}), \sum_{i=9}^{12}[b_i]\varphi_A(Q_{i-8}) \right\rangle.
\end{equation}

We observe that Bob and Alice can use the sets $\mathcal{G}$, associated to $J_{AB}$ and $J_{BA}$, as their shared secret due to the following isomorphism, from which we deduce the correctness of the scheme:

\begin{equation}
J_{AB} = \frac{J_A}{\varphi_A(K_B)} \cong \frac{J_H}{\langle K_A, K_B \rangle} \cong \frac{J_B}{\varphi_B(K_A)} = J_{BA}.
\end{equation}
\section{Genus $2$ oblivious transfer}

In this section we introduce the extension to genus $2$ curves of the proposal for a supersingular isogeny oblivious transfer in \cite{Barreto2018}. We will denote the scheme as $2$-SIOT.

\subsection{Set-up}

Let us consider a prime $p = 2^n3^m f - 1$ for $f\in\Z$ small and exponents $n, m\in\Z$ such that $2^n\approx 3^m$. Fix a supersingular hyperelliptic curve $H$ following the criteria of Section \ref{2sidh} and use a random sequence of Richelot isogenies to get a random principally polarized supersingular abelian surface. Finally we consider generating sets $\{P_1, P_2, P_3, P_4\}$ and $\{Q_1, Q_2, Q_3, Q_4\}$ for $J_H[2^n]$ and $J_H[3^m]$ respectively.

\subsection{First round}

In this first stage Alice receives to messages $m_0, m_1$ from the set of messages $\mathcal{M}$. She will have to select secret scalars $\{a_i\}_{i = 1,\dots , 12}$, following Section 3.2 \cite{Flynn2019}, together with an isogeny $\varphi_A : J_H \to J_A$ whose kernel is generated by

\begin{equation}
K_A = \left\langle \sum_{i=1}^4[a_i]P_i, \sum_{i=5}^8[a_i]P_{i-4}, \sum_{i=9}^{12}[a_i]P_{i-8} \right\rangle.
\end{equation}

Furthermore Alice computes $\{\varphi_A(Q_j)\}_{j = 1, \dots, 4}$ and $\mathcal{G}(J_A)$ in order to define her public key $pk_A = (\mathcal{G}(J_A),\{\varphi_A(Q_j)\}_{j = 1, \dots, 4})$ and sends it to Bob.

At the same time Bob gets as input a bit $\beta\in\{0,1\}$. He is required to select a set of secret scalars $\{b_i\}_{i = 1,\dots , 12}$, according to Section 3.2 \cite{Flynn2019}, together with an isogeny $\varphi_B : J_H \to J_B$ whose kernel is generated by

\begin{equation}
K_B = \left\langle \sum_{i=1}^4[b_i]Q_i, \sum_{i=5}^8[b_i]Q_{i-4}, \sum_{i=9}^{12}[b_i]Q_{i-8} \right\rangle.
\end{equation}

After receiving $pk_A$, Bob checks if $\{\varphi_A(Q_j)\}_{j=1,2,3,4}\in J_A[3^m]$. If not, $pk_A$ is rejected. Otherwise Bob computes elements $U_1, U_2, U_3, U_4 \in J_B[2^n]$ in order to define $R_j = \varphi_B(P_j) - \beta U_j$ for $j \in\{1, 2, 3, 4\}$. Bob defines his public key as $pk_B = (\mathcal{G}(J_B), \{R_j\}_{j = 1, 2, 3, 4})$. This key is sent to Alice.

\subsection{Second round}

After receiving Bob's public key, Alice checks if $\{R_j\}_{j=1,2,3,4} \in J_B[2^n]$. As in Bob's situation, if the condition does not hold, Bob's public key is rejected. Otherwise Alice computes the set, for $k = 0, 1$:

\begin{equation}
\begin{split}
K_{BA_k} = \left\{ \sum_{i=1}^4[a_i]\varphi_B(P_i)-\beta U_i + kU_i, \sum_{i=5}^8[a_i]\varphi_B(P_{i-4})-\beta U_{i-4} + kU_{i-4},\right. \\
\left. {}\sum_{i=9}^{12}[a_i]\varphi_B(P_{i-8})-\beta U_{i-8} + kU_{i-8} \right\},
\end{split}
\end{equation}

which generates the kernel for an isogeny $\varphi'_{A_k}:J_B \to J_{BA_k}$. Let us denote with $\mathcal{G}_k$ the set $\mathcal{G}$ associated to $J_{BA_k}$. Finally, given a symmetric encryption scheme $(Enc, Dec)$, Alice computes $c_k := Enc(\mathcal{G}_k, m_k)$ for $k = 0, 1$ and sends the pair $(c_0, c_1)$ to Bob.

On the other side, Bob computes the set

\begin{equation}
K_{AB} = \left\{ \sum_{i=1}^4[b_i]\varphi_A(Q_i), \sum_{i=5}^8[b_i]\varphi_A(Q_{i-4}), \sum_{i=9}^{12}[b_i]\varphi_A(Q_{i-8}) \right\},
\end{equation}

which generates the kernel of an isogeny $\varphi'_B: J_A\to J_{AB}$. Bob will compute $\mathcal{G}(J_{AB}):= \mathcal{G}_\beta$ which will allow him to perform $m_\beta = Dec(\mathcal{G}_\beta, c_\beta)$.

Below follows the pseudocode view of the $2$-SIOT protocol:
\newpage

\begin{center}
\fbox{%
\pseudocode{
\textbf{Set-up} \<\< \\ [0.1\baselineskip][\hline]
\<\< \\[-0.5\baselineskip]
\text{Prime }p = 2^n3^mf-1 \\
f,n,m\in\Z \text{ s.t. } f\text{ small, } 2^n\approx 3^m \\
H \text{ supersingular hyperelliptic curve} \\
\{P_1,P_2,P_3,P_4\}\text{ gen. } J_H[2^n] \\
\{Q_1,Q_2,Q_3,Q_4\}\text{ gen. } J_H[3^m] \\ [0.5\baselineskip][\hline]
\<\< \\[-0.75\baselineskip]
\textbf{Alice} \<\< \textbf{Bob} \\ [0.1\baselineskip][\hline]
\<\< \\[-0.5\baselineskip]
\text{Input: } m_0, m_1\in\mathcal{M} \<\< \text{Input: }\beta\in\{0,1\} \\
\text{Output: none} \<\< \text{Output: } m_\beta \\
\text{Secret: }\{a_i\}_{i\in\{1,\dots, 12\}} \<\< \text{Secret: }\{b_i\}_{i\in\{1,\dots, 12\}} \\
\text{Isogeny: }\varphi_A: J_H \to J_A \<\< \text{Isogeny: }\varphi_B: J_H \to J_B \\
\left\langle \sum_{i=1}^4[a_i]P_i, \sum_{i=5}^8[a_i]P_{i-4}, \sum_{i=9}^{12}[a_i]P_{i-8} \right\rangle \<\< \left\langle \sum_{i=1}^4[b_i]Q_i, \sum_{i=5}^8[b_i]Q_{i-4}, \sum_{i=9}^{12}[b_i]Q_{i-8} \right\rangle \\
\{\varphi_A(Q_j)\}_{j\in\{1,2,3,4\}} \<\< \{\varphi_B(P_j)\}_{j\in\{1,2,3,4\}} \\
\mathcal{G}(J_A) \<\< \mathcal{G}(J_B) \\
pk_A = (\mathcal{G}(J_A), \{\varphi_A(Q_j)\}_{j=1,2,3,4}) \<\< \\
\< \sendmessageright*[1cm]{pk_A} \< \\
\<\< \{\varphi_A(Q_j)\}_{j=1,2,3,4}\notin J_A[3^m] \Rightarrow \text{Reject} \\
\<\< \{U_j\}_{j=1,2,3,4}\in J_B[2^n] \\
\<\< R_j = \varphi_B(P_j) - \beta U_j,\, j\in\{1,2,3,4\} \\ 
\<\< pk_B = (\mathcal{G}(J_B), \{R_j\}_{j=1,2,3,4}) \\
\< \sendmessageleft*[1cm]{pk_B} \< \\
\{R_j\}_{j=1,2,3,4}\notin J_B[2^n] \Rightarrow \text{Reject} \<\< \\
k\in\{0,1\} \<\< \\
T_{i,k} := [a_i]\varphi_B(P_i) - \beta U_i+kU_i \\
\left\langle \sum_{i=1}^4 T_{i,k}, \sum_{i=5}^8 T_{i-4,k}, \sum_{i=9}^{12} T_{i-8,k} \right\rangle \<\< \\
\varphi'_{A_k} : J_B \to J_{BA_k} \<\< \\
\mathcal{G}_k = \mathcal{G}(J_{BA_k}) \<\< \\
c_k = Enc(\mathcal{G}_k, m_k),\,\forall k\in\{0,1\} \<\< \\
\< \sendmessageright*[1cm]{(c_0, c_1)} \< \\
\<\< \varphi_j = \varphi_A(Q_j),\, j\in\{1,2,3,4\} \\
\<\< \left\langle \sum_{i=1}^4[b_i]\varphi_i, \sum_{i=5}^8[b_i]\varphi_{i-4}, \sum_{i=9}^{12}[b_i]\varphi_{i-8} \right\rangle \\
\<\< \varphi'_B: J_A \to J_{AB} \\
\<\< \mathcal{G}_\beta = \mathcal{G}(J_{AB}) \\
\<\< m_\beta = Dec(\mathcal{G}_\beta, c_\beta)
}
}
\end{center}
\section{Correctness and privacy}

This section is devoted to check the correctness, the privacy and the indistinguishability of the $2$-SIOT proposal. The reader will realise that many of the arguments used in \cite{Barreto2018} can be extended to our setting with a few but tedious computations. The first part is focused on the computational problems which are assumed to be hard in a quantum setting. Correctness is immediately checked in a second subsection and we finish by justifying the privacy and the indistinguishability of our proposal.

\subsection{Background on two-party computation}

A map $\varepsilon:\N \to [0,+\infty)$ is negligible if for every positive polynomial $Q\in\F_p[x]$ and for every $n\in\N$ the following condition holds: $\varepsilon(n) < (Q(n))^{-1}$.

A probability ensemble $X = \{X(n,a)\}_{n\in N, a\in\{0,1\}^\star}$ is a sequence of random variables. The variable $n$ represents the security parameter, whereas $a$ represents the inputs of each party.

Two distribution ensembles $X$ and $Y$ are computationally indistinguishable, and we write it as $X\stackrel{c}{\equiv}Y$, if for every non-uniform polynomial-time algorithm $D$ there exists a negligible map $\varepsilon$ such that:

\begin{equation}
|P(D(X(n,a)) = 1) - P(D(Y(n,a)) = 1)| \leq \varepsilon(n), \forall n\in \N, \forall a\in \{0,1\}^\star.
\end{equation}

The following concepts are well known (Section 7.2 \cite{Goldreich2004}):

A two-party problem can be cast by defining a random process that maps pairs of inputs to pairs of outputs. Such a process is called a functionality and is denoted by a map of the form $F: \{0,1\}^\star \times \{0,1\}^\star \to \{0,1\}^\star \times \{0,1\}^\star$, with $F(x,y) = (F_1(x,y), F_2(x,y))$. Here $x$ denotes the input of first party, and $y$ denotes the input of the second party. After evaluating $F$, the first party should receive $F_1(x,y)$ and the second party should receive $F_2(x,y)$.

Let $F = (F_1, F_2)$ be a functionality and $\pi$ a two-party protocol for computing $F$. The view of the first (resp. second) party during an execution of $\pi$ on $(x,y)$ is denoted by $view_i(x,y)$ is $(x, r, m_1,\dots, m_t)$ (resp. $(y, r, m_1,\dots, m_t)$), where $r$ is the outcome of the first (resp. second) party's internal coin tosses and $m_j$ represents the $j$-th message it has received.

A concept particularly elusive is that of malicious party. We follow Section 7.2.3 \cite{Goldreich2004} in order to discuss some ideas about it. 

When we are working in a malicious setting, it is important to keep in mind that there is no way to force parties to participate in the protocol. One malicious behaviour could simply be not participating in the protocol. This implies, somehow, that if one of the parties is participating against its will, another admitted behaviour in a malicious setting could be aborting at any time, even once one of the parties gets the expected output. Another important fact to bare in mind is that in a malicious setting, it is not possible to determine the inputs of a malicious party.

In general, when talking about malicious parties, we admit the following situations:
\begin{enumerate}
\item Parties refusing to participate in the protocol, once it has been started.
\item Parties substituting their local input.
\item Parties aborting the protocol prematurely.
\end{enumerate}

\subsection{Supersingular isogeny problems}

It is important to observe that the security analysis of the genus $1$ problems extends to our setting easily and so, those interested in further details are invited to read \cite{Galbraith2018}. 

Keeping the notations used so far, let us consider a supersingular hyperelliptic curve $H$ over $\F_{p^2}$ with genus $2$ together with independent bases $\{P_i\}_{i=1,\dots 4}$ and $\{Q_i\}_{i=1,\dots 4}$ for $J_H[2^n]$ and $J_H[3^m]$ respectively.

\begin{problem}[DSSI: Decisional supersingular isogeny problem]\label{dssi}
Let $H'$ be another supersingular hyperelliptic curve over $\F_{p^2}$ with genus $2$. Decide whether $J_{H'}$ is $(3^m, 3^m)$-isogenous to $J_H$.
\end{problem} 

\begin{problem}[CSSI: Computational supersingular isogeny problem]\label{cssi}
Let $H_A$ be a supersingular hyperelliptic curve over $\F_{p^2}$ with genus $2$. Let $\varphi_A: J_H \to J_A$ be an isogeny whose kernel is generated by

\begin{equation}
K_A = \left\{ \sum_{i=1}^4[a_i]P_i, \sum_{i=5}^8[a_i]P_{i-4}, \sum_{i=9}^{12}[a_i]P_{i-8} \right\}
\end{equation}

for some $a_i\in\F_{2^n}$. Given $J_A$ and $\{\varphi_A(Q_j\}_{j=1,\dots, 4}$, find generators for $K_A$.
\end{problem}

\begin{problem}[SSCDH: Supersingular computational Diffie-Hellman]\label{sscdh}
Let us consider $\varphi_A: J_H \to J_A$ an isogeny whose kernel is generated by

\begin{equation}
K_A = \left\{ \sum_{i=1}^4[a_i]P_i, \sum_{i=5}^8[a_i]P_{i-4}, \sum_{i=9}^{12}[a_i]P_{i-8} \right\},
\end{equation}

for some $a_i\in\F_{2^n}$ and let $\varphi_B:J_H\to J_B$ whose kernel is generated by

\begin{equation}
K_B = \left\{ \sum_{i=1}^4[b_i]Q_i, \sum_{i=5}^8[b_i]Q_{i-4}, \sum_{i=9}^{12}[b_i]Q_{i-8} \right\}
\end{equation}

for some $b_i\in\F_{3^m}$. Given $\{\varphi_A(Q_j), \varphi_B(P_j)\}_{j=1,\dots, 4}$ and the jacobians $J_A, J_B$, find the set $\mathcal{G}$ associated to $\frac{J_H}{\langle K_A, K_B \rangle}$.
\end{problem}

\subsection{Correctness}

The correctness of this proposal relies on the fact that both parties are honest, that is just one of the messages $m_\beta = Dec(\mathcal{G}_\beta, c_\beta)$ should be decrypted. We need to prove that if the bit $\beta$ chosen by Bob agrees with the element $k$ chosen by Alice, then it is possible to share a secret. This follows from

\begin{equation}
J_{AB} \cong \frac{J_A}{K_{AB}} \cong \frac{J_H}{\langle K_A, K_B \rangle} \cong J_{BA},
\end{equation}

therefore $\mathcal{G}(J_{AB}) = \mathcal{G}(J_{BA_k})$.

\subsection{Privacy}

Let us denote with $\mathfrak{Alice}$ and with $\mathfrak{Bob}$ the malicious alter egos of Alice and Bob, respectively. In order to prove the privacity of the proposal, we use the following definition:

\begin{definition}[Definition 2.6.1 \cite{Hazay2010}]
A two-message two-party probabilistic polynomial-time protocol is private oblivious transfer if the following holds:

\begin{enumerate}
\item Non-triviality: If Alice and Bob follow the protocol, then after an execution in which Alice has for input any pair of $m_0,m_1\in\mathcal{M}$, and Bob has for input a bit $b\in\{0,1\}$, then the output of Bob is $m_b$.

\item Privacy for $\mathfrak{Alice}$:\label{mean_alice} For every non-uniform polynomial-time $\mathfrak{Alice}$ and every auxiliary input $z\in\{0,1\}^\star$, the following condition holds:
\begin{multline*}
\{view_{\mathfrak{Alice}}(\mathfrak{Alice}(1^n,z), \text{Bob}(1^n,0))\}_{n\in\N} \stackrel{c}{\equiv} \\
\{view_{\mathfrak{Alice}}(\mathfrak{Alice}(1^n,z), \text{Bob}(1^n,1))\}_{n\in\N}.
\end{multline*}

\item Privacy for $\mathfrak{Bob}$:\label{mean_bob} For every non-uniform deterministic polynomial-time $\mathfrak{Bob}$, every auxiliary input $z\in\{0,1\}^\star$ and ever $m_0, m_1, m\in\{0,1\}^\star$ such that $|m_0| = |m_1| = |m|$ one of the conditions below hold:

\begin{multline*}
\{view_{\mathfrak{Bob}}(\text{Alice}(1^n,(m_0,m_1)), \mathfrak{Bob}(1^n,z))\}_{n\in\N} \stackrel{c}{\equiv} \\
\{view_{\mathfrak{Bob}}(\text{Alice}(1^n,(m_0,m)), \mathfrak{Bob}(1^n,z))\}_{n\in\N}.
\end{multline*}
or
\begin{multline*}
\{view_{\mathfrak{Bob}}(\text{Alice}(1^n,(m_0,m_1)), \mathfrak{Bob}(1^n,z))\}_{n\in\N} \stackrel{c}{\equiv} \\
\{view_{\mathfrak{Bob}}(\text{Alice}(1^n,(m,m_1)), \mathfrak{Bob}(1^n,z))\}_{n\in\N}.
\end{multline*}
\end{enumerate}
\end{definition}

The 2-SIOT protocol is non-trivial because Bob can compute the jacobian $J_{AB}$ from $pk_A$ and obtain $\mathcal{G}_\beta$, thus $m_\beta = Dec(\mathcal{G}, c_\beta)$ is such that $\mathcal{G}(J_{AB}) = \mathcal{G}_\beta$ and $\beta$ is a unique element secretly chosen by Bob.

We observe that conditions \ref{mean_alice} and \ref{mean_bob} in the previous definition follow using the arguments in \cite{Barreto2018}. In particular, we have the following analogous results:

\begin{lemma}[Lemma 1 \cite{Barreto2018}]
Alice on input $pk_B$ cannot guess the bit $\beta$ with probability greater that $\frac{1}{2} + \varepsilon(n)$, for some negligible function $\varepsilon(n)$ and $n\in\N$.
\end{lemma}

\begin{proof}
The proof of this lemma is quite technical and can be easily extended to our context after carefully checking the computations. Those interested in the details are welcome to check them in \cite{Barreto2018}.
\end{proof}

\begin{lemma}[cf. Lemma 2 \cite{Barreto2018}]
Assuming that SSCDH (Problem \ref{sscdh}) is hard, Bob cannot compute two distinct sets $\mathcal{G}_0$ and $\mathcal{G}_1$.
\end{lemma}

\begin{proof}
This result follows immediately by reducing to absurd: if we assume that Bob is able to compute two distinct sets $\mathcal{G}_0$ and $\mathcal{G}_1$ then he would violate SSCDH (Problem \ref{sscdh}). 
\end{proof}

\begin{theorem}[cf. Theorem 3.4.1 \cite{Barreto2018}]
Assuming that the decisional Diffie-Hellman problem and the supersingular computational Diffie-Hellman problem are hard for $J_H(\F_{p^2})$, the 2-SIOT proposal is private oblivious transfer.
\end{theorem}

\subsection{Indistinguishability}

Let us assume a setting in which $\mathfrak{Alice}$ receives $pk_B$ from Bob. A priori, $\mathfrak{Alice}$ is not able to decide whether she received $(\mathcal{G}(J_B),\{\varphi_B(P_j)\}_{j=1,2,3,4})$ or if she received $(\mathcal{G}(J_B),\{\varphi_B(P_j)-U_j\}_{j=1,2,3,4})$, that is: $\mathfrak{Alice}$ is not able to decide if Bob got $\beta = 0$ or $\beta = 1$. In order to overcome this problem, $\mathfrak{Alice}$ could use the Weil pairing.

Let us assume that all pairings have order $2^n$. The points $\varphi_B(P_j)$ satisfy $\langle \varphi_B(P_i), \varphi_B(P_j) \rangle_{2^n} = (\langle P_i, P_j \rangle_{2^n})^{3^m}$ for each $1\leq i < j \leq 4$. If the equality does not hold for both $(\mathcal{G}(J_B),\{\varphi_B(P_i)\}_i)$ and $(\mathcal{G}(J_B),\{\varphi_B(P_i)-U_i\}_i)$, then $\mathfrak{Alice}$ would be able to know the key used by Bob.

Therefore, as $\mathfrak{Alice}$ can add any multiple of $\{U_j\}_{j=1,2,3,4}$ to $\{\varphi_B(P_j)\}_{j=1,2,3,4}$ and look for a mismatch, we need to have\footnote{Note the abuse of notation in the Weil pairing.}, $\forall i,j\in\{1,2,3,4\}$ and $\forall \lambda\in \mathbb{F}_{2^n}$:

\begin{equation}\label{weilproduct}
\left\langle\varphi_B(P_i) + \lambda U_i, \varphi_B(P_j) + \lambda U_j\right\rangle = \prod_{i < j}\langle\varphi_B(P_i), \varphi_B(P_j)\rangle.
\end{equation}

For each $i\in\{1,2,3,4\}$, we write $\varphi_i := \varphi_B(P_i)$ and recall that as $U_i\in J_B[2^n]$, then 

\begin{equation}
U_i = \alpha_i\varphi_1 + \beta_i\varphi_2 + \gamma_i\varphi_3 + \delta_i\varphi_4.
\end{equation}

For the sake of simplicity, we consider the case $i = 1, j = 2$. Then condition (\ref{weilproduct}) implies:

\begin{equation*}
\begin{split}
& \langle\varphi_1 + \lambda U_1, \varphi_2 + \lambda U_2\rangle \\
& = \langle \varphi_1 + \lambda(\alpha_1 \varphi_1 + \beta_1 \varphi_2 + \gamma_1 \varphi_3 + \delta_1 \varphi_4), \varphi_2 + \lambda(\alpha_2 \varphi_1 + \beta_2 \varphi_2 + \gamma_2 \varphi_3 + \delta_2 \varphi_4)\rangle \\
& = \langle(1+\lambda\alpha_1)\varphi_1, (1+\lambda\beta_2)\varphi_2\rangle \langle(1+\lambda\alpha_1)\varphi_1, \lambda\gamma_2\varphi_3\rangle \langle(1+\lambda\alpha_1)\varphi_1,\lambda\delta_2\varphi_2\rangle \\
& \cdot \langle\lambda\beta_1\varphi_2, \lambda\alpha_2\varphi_1\rangle \langle\lambda\beta_1\varphi_2,\lambda\gamma_2\varphi_3\rangle \langle\lambda\beta_1\varphi_2,\lambda\delta_2\varphi_4\rangle \langle\lambda\gamma_1\varphi_3, (1+\lambda\beta_2)\varphi_2\rangle \\
& \cdot \langle\lambda\gamma_1\varphi_3,\lambda\alpha_2\varphi_1\rangle \langle\lambda\gamma_1\varphi_3,\lambda\delta_2\varphi_4\rangle \langle\lambda\delta_1\varphi_4,(1+\lambda\beta_2)\varphi_2\rangle \langle\lambda\delta_1\varphi_4,\lambda\alpha_2\varphi_1\rangle \\
& \cdot \langle\lambda\delta_1\varphi_4, \lambda\gamma_2\varphi_3\rangle \\
& = \langle\varphi_1,\varphi_2)^{(1+\lambda\alpha_1)(1+\lambda\beta_2)}\langle\varphi_1,\varphi_3\rangle^{(1+\lambda\alpha_1)\lambda\gamma_2}\langle\varphi_1,\varphi_4\rangle^{(1+\lambda\alpha_1)\lambda\delta_2} \\
& \cdot \langle\varphi_2,\varphi_1\rangle^{\lambda^2\beta_1\alpha_2}\langle\varphi_2,\varphi_3\rangle^{\lambda^2\beta_1\gamma_2}\langle\varphi_2,\varphi_4\rangle^{\lambda^2\beta_1\delta_2}\langle\varphi_3,\varphi_2\rangle^{\lambda\gamma_1(1+\lambda\beta_2)} \\
& \cdot \langle\varphi_3,\varphi_1\rangle^{\lambda^2\gamma_1\alpha_2}\langle\varphi_3,\varphi_4\rangle^{\lambda^2\gamma_1\delta_2}\langle\varphi_4,\varphi_2\rangle^{\lambda\delta_1(1+\lambda\beta_2)}\langle\varphi_4,\varphi_1\rangle^{\lambda^2\gamma_1\alpha_2}\langle\varphi_4,\varphi_3\rangle^{\lambda^2\delta_1\gamma_2} \\
& = \langle\varphi_1, \varphi_2\rangle^{\lambda^2(\alpha_1\beta_2-\alpha_2\beta_1)+\lambda(\beta_2+\alpha_1)+1}\langle\varphi_1,\varphi_3\rangle^{\lambda\alpha_2\gamma_1+\lambda(\gamma_2 + \alpha_2\gamma_1)} \\
& \cdot \langle\varphi_1,\varphi_4\rangle^{\lambda^2(\gamma_2\alpha_1-\gamma_1\alpha_2)+\lambda\delta_2} \langle\varphi_2,\varphi_3\rangle^{\lambda^2(\beta_1\gamma_2-\beta_2\gamma_1)-\lambda\gamma_1}\langle\varphi_2,\varphi_4\rangle^{\lambda^2(\beta_1\delta_2-\beta_2\delta_1)-\lambda\delta_2} \\
& \cdot \langle\varphi_3, \varphi_4\rangle^{\lambda^2(\gamma_1\delta_2-\delta_1\gamma_2)} = \langle\varphi_1, \varphi_2\rangle \langle\varphi_1, \varphi_3\rangle \langle\varphi_1, \varphi_4\rangle \langle\varphi_2, \varphi_3\rangle \langle\varphi_2, \varphi_4\rangle \langle\varphi_3, \varphi_4\rangle.
\end{split}
\end{equation*}

Therefore:

\begin{equation}\label{system}
\begin{cases}
\lambda^2(\alpha_1\beta_2 - \alpha_2\beta_1) + \lambda(\beta_2 + \alpha_1) + 1 & = 1\\
-\lambda^2\alpha_2\gamma_1 + \lambda(\alpha_2\gamma_1 + \gamma_2) & = 1 \\
\lambda^2(\delta_2\alpha_1 - \gamma_1\alpha_2) + \lambda\delta_2 & = 1 \\
\lambda^2(\beta_1\gamma_2 - \beta_2\delta_1) - \lambda\gamma_1 & = 1 \\
\lambda^2(\beta_1\delta_2 - \beta_2\delta_1) - \lambda\delta_2 & = 1 \\
\lambda^2(\gamma_1\delta_2 - \delta_1\gamma_2) & = 1
\end{cases}
\end{equation}

Observe that the first condition in (\ref{system}) implies:

\begin{equation}
\lambda^2(\alpha_1\beta_2 - \alpha_2\beta_1) + \lambda(\beta_2 + \alpha_1) = 0 \Leftrightarrow \beta_2^2 + \alpha_2\beta_1 = 0
\end{equation}

After performing the above tedious but straightforward computations for each pair $i,j \in \{1,2,3,4\}$ we are able to state the following result:

\begin{proposition}\label{propomeva}
Let $i,j\in\{1,2,3,4\}$ denote the variable sets $\alpha, \beta, \gamma, \delta$ and the subindices $\{i,j\}$ specify the variable. Then the identity $j^2_j + i_jj_i = 0$ for each pair $\{i,j\}\in \{1,2,3,4\}$ is a necessary condition for $\mathfrak{Alice}$ not to be able to distinguish between $(\mathcal{G}(J_B),\{R_i\}_i)$ or $(\mathcal{G}(J_B),\{\varphi_B(P_i)\}_i)$. 
\end{proposition}

\begin{remark}
As an example, for indices $i = 2, j = 3$ the condition in Proposition \ref{propomeva} becomes $\gamma^2_3 + \beta_3\gamma_2 = 0$. In the case $i = j$ we obtain $\alpha_1 = \beta_2 = \gamma_3 = \delta_4 = 0$.  
\end{remark}
\section{Conclusion}

We have presented a post-quantum protocol for an oblivious transfer based in principally polarized supersingular abelian surfaces combining and adjusting the ideas of \cite{Barreto2018} and \cite{Flynn2019}. 

We have proved that the proposal is private in settings with malicious sender and/or receiver and we also have given necessary conditions for the proposal to be undistinguishable in a setting involving a malicious sender. Furthermore we have presented the mathematical problems that provide the post-quantum security of the proposal. Nevertheless, the analysis of the situation preventing possible decryptions from a malicious receiver remains open.

It is important to bare in mind that this proposal follows the ideas from \cite{Flynn2019}, and therefore it suffers from the same bad habits: although the protocol is implementable, it is too slow to be practical, therefore improvements on the algorithms must be obtained in order to get a functional version of $2$-SIOT.



\begin{thebibliography}{17}
\bibitem{Barreto2018}
Barreto, P., Oliveira, G., Benits, W.: Supersingular Isogeny Oblivious Transfer. arXiv:1805.06589v1 (2018)

\bibitem{Cardona2002}
Cardona, G., Quer, J.: Field of moduli and field of definition for curves of genus 2. arXiv:math/0207015v1 (2002) 

\bibitem{Castryck2019}
Castryck, W., Decru, T., Smith, B.: Hash functions from superspecial genus-2 curves using Richelot isogenies. arXiv:1903.06451v1 (2019)

\bibitem{Charles2009}
Charles, D.X., Lauter, K.E., Goren, E.Z.: Cryptographic
Hash Functions from Expander Graphs. Journal of Cryptology \textbf{22}(1), 93--113 (2009)

\bibitem{Feo2015}
De Feo, L., Jao, D., Pl\^{u}t, J.: Towards quantum-resistant cryptosystems from supersingular elliptic curve isogenies. Journal of Mathematical Cryptology \textbf{8}(3): 209--247 (2015)

\bibitem{Feo2018}
De Feo, L., Galbraith, S.D.: SeaSign: Compact isogeny signatures from class group actions. Cryptology ePrint Archive: Report 2018/824 (2018)

\bibitem{Flynn2019}
Flynn, E.V., Bo Ti, Y.: Genus Two Cryptography. Cryptology ePrint Archive: Report 2019/177 (2019)

\bibitem{Galbraith2012}
Galbraith, S.D.: Mathematics of public key cryptography. 1st edn. Cambridge University Press, United Kingdom (2012)

\bibitem{Galbraith2018}
Galbraith, S.D., Vercauteren, F.: Computational problems
in supersingular elliptic curve isogenies. Quantum Information Processing \textbf{17}(10), 1--22 (2018)

\bibitem{Goldreich2004}
Goldreich, O.: Foundations of Cryptography: Volume 2, Basic Applications. 1st edn. Cambridge University Press, United States of America (2004)

\bibitem{Gonzalez2004}
Gonzalez, J., Guardia, J., Rotger, V.: Abelian Surfaces of $GL_2$-type as Jacobians of Curves. arXiv:math/0409352v1 (2004)

\bibitem{Hazay2010}
Hazay, C., Lindell, Y.: Efficient Secure Two-Party Protocols: Techniques and Constructions. 1st edn. Springer-Verlag, Germany (2010)

\bibitem{Kilian1988}
Kilian, J.: Founding crytpography on oblivious transfer. STOC '88 Proceedings of the twentieth annual ACM symposium on Theory of computing, 20--31 (1988)

\bibitem{Lang2019}
Lang, S.: Abelian Varieties. Reprint edition. Dover Publications, United States of America (2019)

\bibitem{Milne1986}
Milne, J.S.: Abelian Varieties. In: Cornell, G.,  Silverman, J.H. (eds.) Arithmetic Geometry. Springer-Verlag, United States of America (1986)

\bibitem{Oort1973}
Oort, F., Ueno, K.: Principally Polarized Abelian Variaties of Dimension Two or Three are Jacobian Varieties. Aarhus Universitet Preprint Series \textbf{38} (1973)

\bibitem{Tachibana2017}
Tachibana, H., Takashima, K., Takagi, T.: Constructing an efficient hash function from 3-isogenies. JSIAM Letters \textbf{9}, 29--32 (2017)

\bibitem{Takashima2018}
Takashima, K.: Efficient Algorithms for Isogeny Sequences and Their
Cryptographic Applications. In Takagi, T., Wakayama, M., Tanaka, K., Kunihiro, N., Kimoto, K., Duong, D.H. (eds.) Mathematical Modelling for Next-Generation Cryptography. Springer, Singapore (2018)

\end{thebibliography}
\end{document}